\documentclass[aps,pre,twocolumn,groupedaddress]{revtex4-1}

\usepackage{euscript,amsmath,amssymb,amsfonts,graphicx,bm,amsthm,mathtools}

  \usepackage{paralist}
  \usepackage{epstopdf}
  \usepackage{graphics} 
 \usepackage[latin1]{inputenc}
\usepackage[caption=false]{subfig}
\usepackage{soul}

\usepackage{latexsym,amssymb,enumerate,amsmath,amsthm} 
  \usepackage{graphicx} 
  \usepackage{tikz}
  \usepackage{hyperref}
\usepackage{latexsym,amssymb,enumerate,amsmath} 
\usepackage{pgfplots}
\usepackage{soul,xcolor}

\newcommand{\dd}{\textup{d}}
\def\eps{\varepsilon}
\def\E{\mathbb{E}}
\def\P{\mathbb{P}}
\def\R{\mathbb{R}}
\def\black{\color{black}}

\def\b{\mathbf{b}}
\def\c{\mathbf{c}}
\def\u{\mathbf{u}}
\def\v{\mathbf{v}}
\def\w{\mathbf{u}}
\def\f{\mathbf{f}}
\def\q{\mathbf{q}}
\def\p{\mathbf{p}}
\def\n{\mathbf{n}}
\def\r{\mathbf{r}}

\def\D{\prescript{}{0}D_{t}^{1-\alpha}}
\def\DD{\mathcal{D}}

\def\LL{\mathbb{L}}

\def\A{\mathcal{A}}
\def\G{\mathcal{G}}
\def\domain{V}
\def\kon{k_{\textup{on}}}
\def\koff{k_{\textup{off}}}
\def\dist{\textup{d}}

\newtheorem{theorem}{Theorem}
\newtheorem{lemma}[theorem]{Lemma}

\theoremstyle{plain}
\theoremstyle{remark}

\begin{document}


\title[]{\textcolor{black}{Subdiffusion-limited} fractional reaction-subdiffusion equations \textcolor{black}{with affine reactions}: solution, stochastic paths, and applications}



\author{Sean D. Lawley}
\email[]{lawley@math.utah.edu}
\affiliation{University of Utah, Department of Mathematics, Salt Lake City, UT 84112 USA}


\date{\today}

\begin{abstract}
In contrast to normal diffusion, there is no canonical model for reactions between chemical species which move by anomalous subdiffusion. Indeed, the type of mesoscopic equation describing reaction-subdiffusion depends on subtle assumptions about the microscopic behavior of individual molecules. Furthermore, the correspondence between mesoscopic and microscopic models is not well understood. In this paper, we study the subdiffusion-limited model, which is defined by mesoscopic equations with fractional derivatives applied to both the movement and the reaction terms. Assuming that the reaction terms are affine functions, we show that the solution to the fractional system is the expectation of a random time change of the solution to the corresponding integer order system. This result yields a simple and explicit algebraic relationship between the fractional and integer order solutions in Laplace space. We then find the microscopic Langevin description of individual molecules that corresponds to such mesoscopic equations and give a computer simulation method to generate their stochastic trajectories. This analysis identifies some precise microscopic conditions that dictate when this type of mesoscopic model is or is not appropriate. We apply our results to several scenarios in cell biology which, despite the ubiquity of subdiffusion in cellular environments, have been modeled almost exclusively by normal diffusion. Specifically, we consider subdiffusive models of morphogen gradient formation, fluctuating mobility, and fluorescence recovery after photobleaching (FRAP) experiments. We also apply our results to fractional ordinary differential equations.
\end{abstract}


\maketitle

\section{\label{intro}Introduction}

Subdiffusion has been observed in very diverse systems \cite{oliveira2019, klafter2005, sokolov2012, meroz2015} and is especially prevalent in cell biology \cite{hofling2013, barkai2012}. Subdiffusion is defined by the following sublinear growth in the mean-squared displacement of a tracer particle,
\begin{align}\label{sub}
\E\big[\big(Y(t)-Y(0)\big)^{2}\big]
\propto t^{\alpha},\quad \alpha\in(0,1),
\end{align}
where $Y(t)$ is the one-dimensional position of the particle at time $t\ge0$ and $\E$ denotes expectation. 

A number of mathematical models yield the nonlinear phenomenon in \eqref{sub}, including continuous-time random walks, fractional Brownian motion, and random walks on fractal and disordered systems \cite{hofling2013}. The continuous-time random walk model can be used to derive the following fractional diffusion equation \cite{metzler2000},
\begin{align}\label{fde}
\frac{\partial}{\partial t}c(x,t)
=\D K\frac{\partial^{2}}{\partial x^{2}} c(x,t),\quad x\in\R,\,t>0,
\end{align}
for the concentration $c(x,t)$ of some chemical at position $x$ at time $t$. In the mesoscopic description \eqref{fde}, the parameter $K>0$ is the generalized diffusivity (with dimensions $(\text{length})^{2}(\text{time})^{-\alpha}$) and $\D$ is the Riemann-Liouville fractional derivative \cite{samko1993},
\begin{align}\label{rl}
\D \phi(t)
:=\frac{\dd}{\dd t}\int_{0}^{t}\frac{1}{\Gamma(\alpha)(t-t')^{1-\alpha}}\phi(t')\,\dd t',
\end{align}
where $\Gamma(\alpha)$  is the Gamma function.

An important and now longstanding question is how to model reaction kinetics for subdiffusive molecules (see the review \cite{nepomnyashchy2016} and \cite{hornung2005, gafiychuk2008, boon2012, angstmann2013, kosztolowicz2013, hansen2015, straka2015, dossantos2019, zhang2019, li2019}). In contrast to normal diffusion, there is no canonical model for modeling reactions between subdiffusive molecules. Indeed, significantly different forms of reaction-subdiffusion equations have been proposed {\black (see \cite{nepomnyashchy2016} and also the Discussion section below)}, and the structure of these mesoscopic equations depends on subtle assumptions about the microscopic behavior of individual molecules.

The following form of reaction-subdiffusion equations has been proposed for so-called subdiffusion-limited systems \cite{seki2003b, yuste2004b, nepomnyashchy2016},
\begin{align}\label{formm1}
\frac{\partial}{\partial t}\c
=\D\Big(\text{diag}(K_{1},\dots,K_{n})\frac{\partial^{2}}{\partial x^{2}}\c+{\f}(\c)\Big),
\end{align}
where $\c$ is the vector of $n$ chemical concentrations,
\begin{align*}
\c(x,t)=({c_{i}}(x,t))_{i=1}^{n}\in\R^{n},
\end{align*}
with $n$ generalized diffusivities,  $K_{1},\dots,K_{n}$, and
\begin{align*}
{\f}:\R^{n}\mapsto\R^{n}
\end{align*}
describes reactions between the $n$ species. Importantly, the fractional operator $\D$ is applied to both the movement and the reaction terms in the righthand side of \eqref{formm1}. Models of the form \eqref{formm1} have been derived from continuous-time random walks \cite{seki2003b}, particularly those with instantaneous creation and annihilation \cite{henry2006}. Such models have also been proposed to describe the numerical simulations of \cite{yuste2004b}. Similar models have been used to study subdiffusive bimolecular reactions \cite{yuste2004b, kosztolowicz2006, kosztolowicz2008, kosztolowicz2013}, subdiffusive pattern formation \cite{nec2013}, and traveling waves in subdiffusive media \cite{nec2010, nepomnyashchy2013}. We note that \eqref{formm1} is sometimes written with $\frac{\partial}{\partial t}$ replaced by the Caputo derivative and $\D$ replaced by the identity \cite{nepomnyashchy2016}.

Many fundamental questions regarding {\black equations of the form \eqref{formm1}} remain unanswered. What is the solution? How can we investigate stability? What {\black do such equations} imply about the stochastic movement and reactions of single molecules? How can one simulate the stochastic trajectories of {\black such} individual molecules? What are some biophysical implications for a system following {\black such an equation}?

In this paper, we answer these questions {\black in the case that the reaction term $\f(\c)$ is an affine function of the chemical concentrations $\c$}. In particular, we consider fractional equations of the general form
\begin{align}\label{form0}
\frac{\partial}{\partial t}\c
=\DD(\A\c+\r),\quad x\in\domain\subseteq\R^{d},\,t>0.
\end{align}
In \eqref{form0}, $\domain\subseteq\R^{d}$ is a $d$-dimensional spatial domain (if $\domain$ has a boundary, then we also impose boundary conditions) and $\DD$ is the following integro-differential operator,
\begin{align}\label{DD}
\DD \phi(t)
=\frac{\dd}{\dd t}\int_{0}^{t}M(t-t')\phi(t')\,\dd t',
\end{align}
where $M(t)$ is some given memory kernel (notice that \eqref{DD} reduces to \eqref{rl} if $M(t)=\frac{1}{\Gamma(\alpha)t^{1-\alpha}}$). Further, $\textcolor{black}{\r=\r(x)\in\R^{n}}$ is a space-dependent, time-independent vector, and $\A$ is a linear, spatial operator.

The main example that we have in mind is where $\r\equiv0$ and $\A$ is the diffusion-advection-reaction operator,
\begin{align}\label{ex}
\A\c
=(\text{diag}(\LL_{1},\dots,\LL_{n})+R(x))\c
=\begin{pmatrix}
\LL_{1}c_{1}\\
\vdots\\
\LL_{n}c_{n}
\end{pmatrix}+R(x)\c,
\end{align}
where $\textcolor{black}{R(x):\overline{\domain}\mapsto\R^{n\times n}}$ is a space-dependent matrix and $\LL_{1},\dots,\LL_{n}$ are $n$ forward Fokker-Planck operators, each of the form
\begin{align}\label{fpo}
\begin{split}
&\LL_{i}f(x)
:=-\sum_{j=1}^{d}\frac{\partial}{\partial x_{j}}\big[\mu_{j}(x,i)f(x)\big]\\
&\quad+\frac{1}{2}\sum_{j=1}^{d}\sum_{k=1}^{d}\frac{\partial^{2}}{\partial x_{j}\partial x_{k}}
\Big[\big({\sigma}(x,i){\sigma}(x,i)^{\top}\big)_{j,k}f(x)\Big], 
\end{split}
\end{align}
where $\mu(x,i)\in\R^{d}$ is the external force (drift) vector and $
\sigma(x,i)\in\R^{d\times m}$ describes the space-dependence and anisotropy in the diffusivity for each chemical species $i\in\{1,\dots,n\}$. In this case, $R(x)$ describes the reactions between the $n$ chemical species and $\LL_{i}$ describes the movement of the $i$th species. In the absence of reactions, such equations as \eqref{form0}-\eqref{fpo} are called fractional Fokker-Planck equations \cite{metzler1999}. Notice that \eqref{form0}-\eqref{fpo} becomes \eqref{formm1} if $d=1$, $\domain=\R$, $\mu(x,i)=0$, $\sigma(x,i)=\sqrt{2K_{i}}$, and $\f(\c)=R(x)\c$.

The rest of the paper is organized as follows. In section~\ref{solution}, we show that the solution to \eqref{form0} is
\begin{align}\label{soln}
\c(x,t)
=\E[{\w}(x,S(t))],
\end{align}
where ${\w}(x,s)$ satisfies the corresponding integer order equation (namely \eqref{form0} with $\DD$ replaced by the identity) and $S(t)$ is the inverse of a L\'{e}vy subordinator with Laplace exponent $\Psi(\lambda)$ given by the reciprocal of the Laplace transform of the memory kernel in the integro-differential operator $\DD$ in \eqref{DD},
\begin{align*}
\Psi(\lambda)
=\frac{1}{\widehat{M}(\lambda)},
\end{align*}
where the Laplace transform in time is denoted by
\begin{align*}
\widehat{\phi}(\lambda)
:=\int_{0}^{\infty}e^{-\lambda t}\phi(t)\,\dd t.
\end{align*}
We obtain \eqref{soln} by proving the following algebraic relationship between $\c$ and $\w$ in Laplace space,
\begin{align}\label{solnlp}
\widehat{\c}(x,\lambda)
=\frac{\Psi(\lambda)}{\lambda}\widehat{\w}(x,\Psi(\lambda)).
\end{align}
We also show how \eqref{soln} yields a sufficient condition for {\black linear} stability when the reactions in \eqref{form0} are nonlinear. In section~\ref{stochrep}, we give the stochastic Langevin representation of individual molecules described by \eqref{form0} with $\A$ in \eqref{ex}-\eqref{fpo}. Specifically, we construct a stochastic process whose probability density satisfies \eqref{form0}-\eqref{fpo} when $R(x)$ has a certain probabilistic structure. In this section, we also give a stochastic simulation algorithm to generate realizations of the stochastic process underlying \eqref{form0}. In section~\ref{examples}, we apply our results to some examples of biophysical interest. In particular, we analyze subdiffusive models of protein gradient formation, stochastically switching mobility, and fluorescence recovery after photobleaching (FRAP) experiments. In section~\ref{sectionode}, we apply our results to fractional ordinary differential equations (ODEs). We conclude by discussing related work and future directions.

\section{\label{solution}Exact solution}

In this section, we show that \eqref{soln} satisfies the fractional equations in \eqref{form0} if $\u(x,s)$ satisfies the corresponding integer order equations. The main rigorous result is Theorem~\ref{laplacespace} in section~\ref{setup}, which makes no reference to \eqref{form0}. Instead, Theorem~\ref{laplacespace} is a general result about the Laplace transform of any function subordinated by a continuous, inverse L\'{e}vy subordinator (as in \eqref{soln}), assuming the function satisfies a mild integrability assumption (see \eqref{tonelli}). In section~\ref{imp}, we then show formally how Theorem~\ref{laplacespace} implies that \eqref{soln} satisfies \eqref{form0}. In sections~\ref{bc}-\ref{stability}, we work out some implications of this result.


\subsection{\label{setup}Main theorem}

Let the stochastic process \textcolor{black}{$T=\{T(s)\}_{s\ge0}$} be a L\'{e}vy subordinator. That is, $T$ is a one-dimensional, nondecreasing L\'{e}vy process with $T(0)=0$ \cite{bertoin1996, sato1999}. For each fixed $s>0$, assume that $T(s)$ is a continuous random variable, which means
\begin{align}\label{ctsrv}
\P(T(s)=t)=0,\quad\text{for all }s>0\text{ and }t\ge0.
\end{align}
Let $\Psi(\lambda)$ denote the Laplace exponent of $T$, which means that for all $s\ge0$ and $\lambda\ge0$,
\begin{align}
\E[e^{-\lambda T(s)}]
&=e^{-s\Psi(\lambda)}, \label{le}\\
\Psi(\lambda)
&=b \lambda + \int_{0}^{\infty}(1-e^{-\lambda z})\,\nu(\dd z),\nonumber
\end{align}
where $b\ge0$ {\black is the drift and $\nu$ is the L\'{e}vy measure}. 
Let $S=\{S(t)\}_{t\ge0}$ be the inverse subordinator of $T$,
\begin{align}\label{S}
S(t)
:=\inf\{s>0:T(s)>t\}.
\end{align}
Notice that $S(0)=T(0)=0$ almost surely. Notice also that paths of $S$ {\black are continuous} functions of $t$, since \eqref{ctsrv} implies that paths of $T$ {\black are strictly increasing} functions of $s$. 

\begin{theorem}\label{laplacespace}
Let
\begin{align*}
\w(s)
=(u_{i}(s))_{i=1}^{n}:[0,\infty)\mapsto\R^{n},
\end{align*}
be a given function of time. 
Fix $\lambda>0$ and assume that for each component $i\in\{1,\dots,n\}$,
\begin{align}\label{tonelli}
\int_{0}^{\infty}e^{-\lambda t}\E\big|u_{i}(S(t))\big|\,\dd t<\infty.
\end{align}
If we define {\black $\c(t)
:=\E[\w(S(t))]$ for $t\ge0$}, then
\begin{align*}
\lambda\widehat{\c}(\lambda)
=\Psi(\lambda)\widehat{\w}(\Psi(\lambda)).
\end{align*}
\end{theorem}


\textcolor{black}{The proof of Theorem~\ref{laplacespace} is given in the Appendix.}

\subsection{\label{imp}Fractional equations}

We now use Theorem~\ref{laplacespace} to solve fractional equations. Consider the fractional system,
\begin{align}\label{ceq}
\begin{split}
\frac{\partial}{\partial t}\c
&=\DD\big(\A\c+\r\big),\quad x\in\domain\subseteq\R^{d},\,t>0,\\
\c(x,0)
&=\c_{\text{init}}(x),
\end{split}
\end{align}
where $\domain\subseteq\R^{d}$ is some $d$-dimensional spatial domain and the initial condition $\c_{\text{init}}$ is a given bounded function of space. Assume $\DD$ is the integro-differential operator in \eqref{DD} with memory kernel $M(t)$ defined by its Laplace transform,
\begin{align}\label{mdef}
\widehat{M}(\lambda)
=\frac{1}{\Psi(\lambda)},
\end{align}
and assume $M$ is sufficiently regular so that
\begin{align}\label{Mreg}
\lim_{t\to0+}\int_{0}^{t}M(t')\,\dd t'=0.
\end{align}Assume the operator $\A$ commutes with scalar multiplication, Laplace transforms in time, and the fractional temporal operator $\DD$. That is, assume
\begin{align}
\A\beta\mathbf{w}(x,t)
&=\beta\A\mathbf{w}(x,t),\label{commute1}\\
\widehat{(\A\mathbf{w})}(x,\lambda)
&=\A\widehat{\mathbf{w}}(x,\lambda),\label{commute2}\\
\DD\A\mathbf{w}(x,t)
&=\A\DD\mathbf{w}(x,t),\label{commute3}
\end{align}
for scalar constants $\beta>0$ and functions
\begin{align*}
\mathbf{w}:\overline{V}\times[0,\infty)\mapsto\R^{n}
\end{align*}
in the domain of $\A$. For example, if $\A$ is a sufficiently regular linear differential operator acting on the spatial variable $x$ (as in \eqref{ex}), then \eqref{commute1}-\eqref{commute3} hold. More generally, $\A$ could be a linear integro-differential operator acting on $x$. In addition, $\A$ need not even act on $x$, but could instead simply be a matrix $\A=R\in\R^{n\times n}$, in which case \eqref{ceq} becomes a system of fractional ODEs (see section~\ref{sectionode}).

Suppose $\w(x,s)=({u_{i}}(x,s))_{i=1}^{n}$ satisfies the system of integer order equations corresponding to \eqref{ceq} with the same initial condition,
\begin{align}\label{weq}
\begin{split}
\frac{\partial}{\partial s}\w
&=\A\w+\r,\quad x\in\domain\subseteq\R^{d},\,s>0,\\
\w(x,0)
&=\c_{\text{init}}(x).
\end{split}
\end{align}
Assuming that \eqref{ceq} and \eqref{weq} are sufficiently regular to admit Laplace transformation, we claim that the following definition of $\c(x,t)$ satisfies \eqref{ceq},
\begin{align}\label{defn}
\c(x,t)
:=\E[\w(x,S(t))].
\end{align}

To see this, we work with the Laplace transforms of \eqref{ceq} and \eqref{weq}, which are
\begin{align}
\lambda\widehat{\c}(x,\lambda)
-\c_{\text{init}}(x)
&=\frac{\lambda}{\Psi(\lambda)}\Big[\A\widehat{\c}(x,\lambda)+\frac{\r(x)}{\lambda}\Big],\label{ceql}\\
{\lambda}\widehat{\w}(x,{\lambda})
-\c_{\text{init}}(x)
&=\A\widehat{\w}(x,{\lambda})+\frac{\r(x)}{{\lambda}}.
\label{weql}
\end{align}
In obtaining \eqref{ceql}-\eqref{weql}, we used \eqref{commute1}-\eqref{commute3} and that
\begin{align*}
\widehat{\DD\c}
=\frac{\lambda}{\Psi(\lambda)}\widehat{\c},\quad
\widehat{\DD\r}
=\frac{\lambda}{\Psi(\lambda)}\widehat{\r}
=\frac{\r}{\Psi(\lambda)},
\end{align*}
which follows from the convolution form of $\DD$ in \eqref{DD}, the relation in \eqref{mdef}, and \eqref{Mreg}. Now, it is a straightforward algebra exercise to use \eqref{commute1}-\eqref{commute3} to show that if $\widehat{\u}$ satisfies \eqref{weql} and $\widehat{\c}$ and $\widehat{\u}$ satisfy the following relation,
\begin{align}\label{nicerela}
\lambda\widehat{\c}(x,\lambda)
=\Psi(\lambda)\widehat{\w}(x,\Psi(\lambda)),
\end{align}
then $\widehat{\c}$ satisfies \eqref{ceql}. Of course, \eqref{nicerela} is precisely the relation found in Theorem~\ref{laplacespace} for each fixed $x\in\overline{\domain}$.

Summarizing, if we define $\c$ by \eqref{defn}, then Theorem~\ref{laplacespace} implies that $\c$ and $\u$ satisfy \eqref{nicerela}. Therefore, if $\u$ satisfies the Laplace space equation in \eqref{weql} (which is equivalent to \eqref{weq}), then $\c$ satisfies the Laplace space equation in \eqref{ceql}. But, the Laplace space equation \eqref{ceql} is equivalent to \eqref{ceq}. Hence, $\c$ satisfies \eqref{ceq} as desired.

\subsection{\label{bc}Boundary conditions}

In the case that the spatial domain $\domain\subseteq\R^{d}$ is bounded, we impose boundary conditions. Suppose the solution $\w(x,s)$ to \eqref{weq} satisfies boundary conditions of the form,
\begin{align}\label{bcs}
A(x)\frac{\partial}{\partial\n}\w(x,s)+B(x)\w(x,s)
=\v(x),\quad x\in\partial\domain,
\end{align}
where $\frac{\partial}{\partial\n}$ denotes differentiation with respect to the normal derivative, $A(x),B(x)\in\R^{n\times n}$ are given space-dependent matrices, and $\v(x)\in\R^{n}$ is a given space-dependent vector. Then, it is immediate that $\c(x,t):=\E[\w(x,S(t))]$ satisfies the boundary conditions in \eqref{bcs} assuming sufficient regularity to interchange $\frac{\partial}{\partial\n}$ with $\E$. Similarly, if $\domain\subseteq\R^{d}$ is unbounded, then appropriate growth conditions on $\u$ also apply to $\c$.

\subsection{\label{stability}\textcolor{black}{Steady-states and stability}}

The formula \eqref{defn} relates the fractional order solution $\c$ to the integer order solution ${\w}$. It follows from \eqref{defn} that if $\w$ approaches a finite steady-state,
\begin{align}\label{ltl90}
\w_{\text{ss}}(x)
:=\lim_{s\to\infty}\w(x,s)\in\R^{n},
\end{align}
then $\c$ inherits this same finite steady-state,
\begin{align}\label{ltl91}
\lim_{t\to\infty}\c(x,t)
=\w_{\text{ss}}(x).
\end{align}
To see this, fix $x\in\overline{V}$ and let $\w(x,s)$ be any bounded function of time $s\in[0,\infty)$ satisfying \eqref{ltl90}. Since $S(t)\to\infty$ as $t\to\infty$ with probability one, the Lebesgue dominated convergence theorem yields \eqref{ltl91}.

We emphasize that the limit in \eqref{ltl90} is assumed to be finite, since it is possible for $\w$ to diverge and $\c$ to approach a finite limit (see section~\ref{sectionode} below). {\black Note that a steady-state $\u_{\text{ss}}$ of \eqref{weq} satisfies $\A\u_{\text{ss}}+\r=0$. In the case that $\A$ is the reaction diffusion operator in \eqref{ex}-\eqref{fpo}, the steady-state $\u_{\text{ss}}$ satisfies the spatial differential equation $(\text{diag}(\LL_{1},\dots,\LL_{n})+R(x))\u_{\text{ss}}=-\r$. For a simple example, see section~\ref{morph}.}

{\black 
One consequence of \eqref{ltl91} is that the stability of an integer order equation implies the stability of the corresponding fractional order equation. Interestingly, the converse of this statement is in general false. That is, stability of a fractional equation does not imply stability of the corresponding integer order equation (see section~\ref{sectionode} below). 
}

{\black A second consequence of \eqref{ltl91} is that so-called \emph{linear stability} of integer order equations with nonlinear reactions implies linear stability of fractional equations with nonlinear reactions. Recall that a steady-state of a nonlinear system is said to be linearly stable if the system obtained by linearizing about the steady-state is stable \cite{cross1993, yadav2006}. Consider} the system of fractional equations,
\begin{align}\label{nonlinear}
\frac{\partial}{\partial t}\c
=\DD\Big(\A\c+\f(\c)\Big),
\end{align}
where $\f:\R^{n}\mapsto\R^{n}$ is some nonlinear function of $\c$. 
Suppose that {\black \eqref{nonlinear}} has a steady-state, $\c_{\text{ss}}\in\R^{n}$, which implies
\begin{align}\label{vanish}
\textcolor{black}{\A\c_{\text{ss}}+\f(\c_{\text{ss}})
=0.}
\end{align}
{\black Define $\b(x,t)$ via the relation $\c(x,t)=\c_{\text{ss}}+\eps\b(x,t)$, and assume $\b(x,0)$ is order one and $\eps\ll1$. 
Differentiating $\b(x,t)$, Taylor expanding $\f$ about $\c_{\text{ss}}$, and using \eqref{vanish} yields
\begin{align}
\frac{\partial}{\partial t}\b
=\frac{1}{\eps}\frac{\partial}{\partial t}\c
&=\frac{1}{\eps}\DD\Big(\A(\c_{\text{ss}}+\eps\b)+\f(\c_{\text{ss}}+\eps\b)\Big)\nonumber\\
&=\DD\Big(\A\b+R_{\f}\b\Big)
+\mathcal{O}(\eps),\label{beqn0}
\end{align}
where} $R_{\f}\in\R^{n\times n}$ is the Jacobian of $\f$ evaluated at $\c_{\text{ss}}$. {\black Neglecting the order $\eps$ term in \eqref{beqn0} yields the leading order linear equation,
\begin{align}\label{beqn}
\frac{\partial}{\partial t}\b_{0}
&=\DD\Big(\A\b_{0}+R_{\f}\b_{0}\Big).
\end{align}
The steady-state $\c_{\text{ss}}$ is said to be \emph{linearly stable} if $\lim_{t\to\infty}\b_{0}=0$ \cite{cross1993}. Note that linear stability does not always imply stability of the nonlinear system \eqref{nonlinear}, meaning $\lim_{t\to\infty}\b_{0}=0$ may not imply $\lim_{t\to\infty}\c=\c_{\text{ss}}$ \cite{normand1977}.}

Since \eqref{beqn} is linear, the solution is $\textcolor{black}{\b_{0}}(x,t)=\E[\w(x,S(t))]$ where $\w(x,s)$ satisfies \eqref{beqn} with $\DD$ replaced by the identity. Hence, if $\lim_{s\to\infty}\w(x,s)=0$, then {\black \eqref{ltl91} implies} $\lim_{t\to\infty}\textcolor{black}{\b_{0}}(x,t)=0$, and thus the steady-state, $\c_{\text{ss}}$, for the fractional nonlinear equation {\black \eqref{nonlinear}} is {\black linearly} stable. {\black But}, the equation for $\w$ is {\black merely} the linearization of \eqref{nonlinear} with $\DD$ replaced by the identity. Therefore, we conclude that linear stability of a nonlinear, integer order equation implies linear stability of the corresponding nonlinear, fractional order equation. However, we {\black again} caution that stability of a fractional equation does not imply stability of the corresponding integer order equation (see section~\ref{sectionode} below). Summarizing, {\black linear} stability of an integer order equation is a sufficient (but not necessary) condition for {\black linear} stability of the corresponding fractional equation.

\section{\label{stochrep}Stochastic representation}

In this section, we construct a stochastic process whose probability density satisfies \eqref{form0} in the case that $\r\equiv0$ and the operator $\A$ is given by \eqref{ex} and the reaction matrix $R(x)$ has a certain probabilistic structure. In particular, we assume that for each $x\in\overline{\domain}\subseteq\R^{d}$, the matrix $R(x)$ has nonnegative off-diagonal entries (meaning $R(x)$ is a so-called Metzler matrix \cite{farina2011}) and the diagonal entries are such that each column of $R(x)$ sums to zero.

\subsection{Internal Markov process}

{\black In order to construct a non-Markovian stochastic process $(Y(t),J(t))$ whose probability density satisfies a fractional equation, we first construct a Markov process $(X(s),I(s))$. We then define $(Y(t),J(t))$ as a subordination (i.e.\ a random time change) of $(X(s),I(s))$.} 

{\black Suppose $\{X(s)\}_{s\ge0}$ satisfies the stochastic differential equation (SDE),
\begin{align}\label{sde}
\dd X(s)
=\mu(X(s),I(s))\,\dd s
+\sigma(X(s),I(s))\,\dd W(s),
\end{align}
where $\{W(s)\}_{s\ge0}$ is a standard $m$-dimensional Brownian motion {\black and $\mu$ and $\sigma$ are as in \eqref{fpo}}. Notice that the SDE \eqref{sde} depends on $I(s)$. We suppose $\{I(s)\}_{s\ge0}$ is a continuous-time jump process on $\{1,\dots,n\}$ that jumps from state $I(s)=i$ to state $j\neq i$ at rate $(R(X(s)))_{j,i}\ge0$ at time $s\ge0$.}

In words, $X(s)$ {\black follows} an SDE whose righthand side switches according to the jump process $I(s)$, and the jump rates of $I(s)$ may depend on the position $X(s)$. {\black To illustrate, if the initial state is $I(0)=i$, then $X(s)$ diffuses with drift $\mu(X(s),i)$ and diffusivity $\tfrac{1}{2}\sigma(X(s),i)^{2}$ until $I$ jumps to a new state $j\neq i$. Then, $X(s)$ diffuses with drift $\mu(X(s),j)$ and diffusivity $\tfrac{1}{2}\sigma(X(s),j)^{2}$ until $I$ jumps again, etc.} The process $(X(s),I(s))$ is sometimes called a hybrid switching diffusion \cite{yinbook}. The word ``hybrid'' is used because the process combines the continuous dynamics of $X(s)$ with the discrete dynamics of $I(s)$. {\black For a specific example of $(X(s),I(s))$, see section~\ref{exswitch} below.}

{\black The precise mathematical definition of $(X(s),I(s))$ is in terms of its infinitesimal generator. Precisely, $\{(X(s),I(s))\}_{s\ge0}$ is a }Markov process on the state space $\overline{\domain}\times\{1,\dots,n\}$ with generator $\G$ defined by
\begin{align*}
\G f(x,i)
=\LL_{i}^{*}f(x,i)+\sum_{j=1}^{n}(R^{\top}(x))_{i,j}f(x,j),
\end{align*}
where $\LL_{i}^{*}$ is the formal adjoint of $\LL_{i}$ in \eqref{fpo} and $R^{\top}$ is the transpose of $R$, meaning $(R^{\top}(x))_{i,j}=(R(x))_{j,i}$. The generator $\G$ acts on functions $f(x,i):\overline{\domain}\times\{1,\dots,n\}\mapsto\R$ which are twice-continuously differentiable in $x$. In the language of Markov processes, $\G$ is the backward operator corresponding to the forward operator $\A$.

Let $\q_{i}(x,s)$ be the probability density that $X(s)=x$ and $I(s)=i$. If we define the vector $\q(x,s)=(\q_{i}(x,s))_{i=1}^{n}\in\R^{n}$, then the forward Fokker-Planck equation for $\q$ is
\begin{align}\label{qpde}
\frac{\partial}{\partial s}\q
=\A\q,\quad x\in\domain\subseteq\R^{d},\,s>0.
\end{align}
In the case that $\domain$ has a boundary, boundary conditions are imposed on $\q$ corresponding to the assumed behavior of $X(s)$ on the boundary. For example, if $X(s)$ reflects from some portion of the boundary $\partial\domain_{0}\subseteq\partial\domain$ when $I(s)=i$, then
\begin{align*}
\frac{\partial}{\partial\n}\q_{i}(x,s)=0,\quad x\in\partial\domain_{0}.
\end{align*}
Alternatively, if $X(s)$ is absorbed at $\partial\domain_{0}$ when $I(s)=i$, then
\begin{align*}
\q_{i}(x,s)=0,\quad x\in\partial\domain_{0}.
\end{align*}

\subsection{Random time changed process}

Let $\{S(t)\}_{t\ge0}$ be the inverse subordinator in \eqref{S} that is taken to be independent of $\{(X(s),I(s))\}_{s\ge0}$. Define the stochastic process
\begin{align}\label{yj}
\big(Y(t),J(t)\big)
:=\big(X(S(t)),I(S(t))\big),\quad t\ge0.
\end{align}
Let $\p_{i}(x,t)$ be the probability density that $Y(t)=x$ and $J(t)=i$ and define the vector $\p(x,s)=(\p_{i}(x,s))_{i=1}^{n}$. By conditioning on the value of $S(t)$ and using independence, it follows that
\begin{align*}
\p(x,t)
=\E[\q(x,S(t))].
\end{align*}
Therefore, our analysis in section~\ref{solution} yields
\begin{align}\label{ppde}
\frac{\partial}{\partial t}\p
=\DD\A\p,\quad x\in\domain\subseteq\R^{d},\,t>0,
\end{align}
and $\p$ satisfies the same boundary conditions as $\q$.

Summarizing, the (mesoscopic) fractional reaction-subdiffusion equations in \eqref{ppde} describe (microscopic) individual stochastic molecules which evolve according to \eqref{yj}. In particular, $Y(t)$ denotes the spatial position of a particle and $J(t)$ denotes its discrete state. We now investigate the dynamics of $(Y(t),J(t))$ to understand what fractional reaction-diffusion equations of the form \eqref{ppde} imply about the dynamics of individual molecules. 

We see from \eqref{yj} and \eqref{sde} that the particle subdiffuses with dynamics that switch according to its discrete state. In particular, the path of $Y(t)$ follows the path of $X(s)$, but the motion of $Y(t)$ is punctuated by ``pauses'' of the inverse subordinator $S(t)$ (which correspond to jumps of the subordinator $T(s)$, see section~\ref{exswitch}). Analogously, $J(t)$ follows the path of $I(s)$, but $J(t)$ pauses when $S(t)$ pauses. Importantly, notice that $J(t)$ pauses exactly when $Y(t)$ pauses, and therefore $J(t)$ cannot jump when $Y(t)$ is paused. Hence, we obtain one simple microscopic property implied by the mesoscopic equations in \eqref{ppde}.

Next, we investigate the time between jumps of $J(t)$. In the case that $R(x)$ is constant in space, the jump times of $I(s)$ are exactly exponentially distributed. In particular, the time that $I(s)$ spends in state $i$ is an exponential random variable with rate $\lambda_{i}:=\sum_{j\neq i}R_{j,i}$. Letting $\sigma$ denote this exponential time, it follows that $J(t)$ spends time $T(\sigma)$ in state $i$. We thus obtain an additional microscopic property implied by the mesoscopic equations in \eqref{ppde}.

Moreover, we can compute the probability distribution for the sojourn time $T(\sigma)$ in the typical case that the fractional operator is the Riemann-Liouville derivative, $\DD=\D$ in \eqref{rl} with $\alpha\in(0,1)$. In this case, the subordinator $T$ is an $\alpha$-stable subordinator. A direct calculation shows that this random time has the following distribution \cite{pillai1990, meerschaert2011},
\begin{align}\label{ml}
\P(T(\sigma)>t)
=E_{\alpha}(-\lambda_{i}t^{\alpha}),\quad t>0,
\end{align}
where $E_{\alpha}$ is the Mittag-Leffler function,
\begin{align*}
E_{\alpha}(z)
:=\sum_{k=0}^{\infty}\frac{z^{k}}{\Gamma(1+\alpha k)}.
\end{align*}
Hence, a microscopic condition implied by the mesoscopic equations in \eqref{ppde} in this case is that the particle switches states at Mittag-Leffler distributed times described by \eqref{ml}.

\subsection{\label{stochsim}Stochastic simulation}

Having constructed the stochastic process $(Y(t),J(t))$ in \eqref{yj} that corresponds to the fractional equations \eqref{ppde}, we can simulate stochastic paths of this process. This simulation involves two main steps: (i) approximating the path of the internal Markov process $\{(X(s_{k}),I(s_{k}))\}_{k}$ on some internal time mesh $\{s_{k}\}_{k}$, and (ii) approximating the path of the inverse subordinator $\{S(t_{k})\}_{k}$ on some time mesh $\{t_{k}\}_{k}$.

Step (i) is well-studied. For example, see Chapter 5 in \cite{yinbook}. Furthermore, if the transition rate matrix is constant ($R(x)\equiv R$), then step (i) entails merely simulating paths of $I(s)$ (which can be done exactly and efficiently with the Gillespie algorithm \cite{gillespie1977}) and simulating paths of $X(s)$ between jumps of $I(s)$, which can be done with any simulation method for SDEs (see \cite{kloeden2013}).

Step (ii) depends on the particular subordinator $T(s)$ under consideration. In the case that $T(s)$ is an $\alpha$-stable subordinator, Magdziarz et al.\ \cite{magdziarz2007} developed an efficient algorithm for simulating paths of $T(s)$ and $S(t)$. Carnaffan and Kawai \cite{carnaffan2017} developed methods for simulating paths of $T(s)$ and $S(t)$ for the cases that $T(s)$ is a tempered stable subordinator or a gamma subordinator.

Having obtained the simulated values $\{(X(s_{k}),I(s_{k}))\}_{k}$ and $\{S(t_{k})\}_{k}$ by the methods just referenced, one can obtain $X(S(t_{k}))$ from a simple linear interpolation between $X(s_{\overline{k}})$ and $X(s_{\overline{k}+1})$, where the index $\overline{k}$ is chosen so that $s_{\overline{k}}\le S(t_{k})\le s_{\overline{k}+1}$. Similarly, one can set $J(S(t_{k}))=I(s_{\widetilde{k}})$ where $\widetilde{k}$ is the largest index such that $s_{\widetilde{k}}\le S(t_{k})$. We illustrate this method in section~\ref{exswitch} below.

\section{\label{examples}Biophysical applications}

We now apply our results to some biophysical systems which have typically been modeled by normal diffusion.

\subsection{Subdiffusive morphogen gradient formation\label{morph}}

The formation of morphogen gradients, such as the bicoid gradient of \emph{Drosophila}, is often modeled by diffusion away from a localized source and subsequent degradation. The degradation often results from binding to receptors in the cell membrane \cite{porcher2010}. The basic theory can be illustrated with a reaction-diffusion equation \cite{berezhkovskii2010},
\begin{align}\label{gradient}
\frac{\partial}{\partial s}{u}
=D\frac{\partial^{2}}{\partial x^{2}}{u}-k{u},\quad x>0,\,s>0,
\end{align}
modeling the protein (morphogen) concentration ${u}(x,s)$ at position $x$ at time $s$, which diffuses with diffusivity $D>0$ and degrades at rate $k>0$. The protein source can be modeled by specifying a constant flux $\varphi>0$ boundary condition at $x=0$,
\begin{align}\label{source}
-D\frac{\partial}{\partial x}{u}
=\varphi>0,\quad x=0,
\end{align}
and it is assumed that there is no protein initially,
\begin{align}\label{ic}
{u}=0,\quad s=0.
\end{align}
The solution to \eqref{gradient}-\eqref{ic} is \cite{bergmann2007}
\begin{align}\label{timedep}
\begin{split}
u(x,s)
=u_{\text{ss}}(x)\bigg[1&-\frac{1}{2}\text{erfc}\Big(\sqrt{\overline{s}}-\frac{\overline{x}}{\sqrt{4\overline{s}}}\Big)\\
&-\frac{e^{2\overline{x}}}{2}\text{erfc}\Big(\sqrt{\overline{s}}+\frac{\overline{x}}{\sqrt{4\overline{s}}}\Big)\bigg],
\end{split}
\end{align}
where $\overline{x}=(\sqrt{k/D})x$ and $\overline{s}=ks$ are dimensionless space and time variables and the steady-state solution is the decaying exponential,
\begin{align}\label{wss}
u_{\text{ss}}(x)
=\frac{\varphi}{\sqrt{Dk}} e^{-\overline{x}}.
\end{align}

A common tool to characterize the time it takes the time-dependent gradient \eqref{timedep} to approach the steady-state gradient \eqref{wss} is the \emph{accumulation time} \cite{berezhkovskii2010, berezhkovskii2011}. The accumulation time $\tau(x)$ is defined by
\begin{align}\label{tau}
\tau(x)
:=\int_{0}^{\infty}-s\frac{\partial R}{\partial s}(x,s)\,\dd s
=\int_{0}^{\infty}R(x,s)\,\dd s,
\end{align}
where $R(x,s)$ is the local relaxation function which measures the approach of $u(x,s)$ to $u_{\text{ss}}(x)$,
\begin{align}\label{RRR}
R(x,s)
=\frac{u(x,s)-u_{\text{ss}}(x)}{u(x,0)-u_{\text{ss}}(x)}
=1-\frac{u(x,s)}{u_{\text{ss}}(x)}.
\end{align}
The relaxation function $R(x,s)$ is similar to a survival probability, and thus the accumulation time $\tau(x)$ is analogous to a mean first passage time \cite{berezhkovskii2010, berezhkovskii2011}. Using \eqref{timedep}, it is straightforward to calculate that \eqref{tau} is
\begin{align*}
\tau(x)
=\frac{1}{2k}\big(1+(\sqrt{k/D})x\big).
\end{align*}

We can now use the analysis in sections~\ref{solution}-\ref{stochrep} above to investigate how this standard theory is modified if the proteins move subdiffusively and the degradation is subdiffusion-limited. Indeed, since degradation requires that a protein reaches a receptor, it is quite plausible that the degradation could be limited by the subdiffusive proteins. Analogous to \eqref{gradient}-\eqref{source}, the subdiffusive protein concentration $c(x,t)$ now satisfies
\begin{align}\label{second}
\begin{split}
\frac{\partial}{\partial t}c
&=\DD\Big(D\frac{\partial^{2}}{\partial x^{2}}c-kc\Big),\quad x>0,\,t>0,\\
-D\frac{\partial}{\partial x}c
&=\varphi_{0}>0,\quad x=0,\\
c
&=0,\quad t=0,
\end{split}
\end{align}
for some integro-differential operator $\DD$ as in \eqref{DD}. Note that the parameters $D$ and $k$ in \eqref{gradient}-\eqref{source} necessarily differ from the $D$ and $k$ in \eqref{second} (they have different units), but we keep the same notation for simplicity. To solve \eqref{second}, we take the Laplace transform of the time-dependent diffusive solution in \eqref{timedep} 
and use the relation~\eqref{nicerela} of section~\ref{solution} above to obtain the Laplace transform of the solution to \eqref{second},
\begin{align}\label{csubhat}
\begin{split}
\widehat{c}(x,\lambda)
&=\frac{\Psi(\lambda)}{\lambda}\widehat{{u}}(x,\Psi(\lambda))\\
&={u}_{\text{ss}}(x)\frac{\exp(\overline{x}(1-\sqrt{1+\Psi(\lambda)/k}))}{\lambda\sqrt{1+\Psi(\lambda)/k}},
\end{split}
\end{align}
where $\Psi(\lambda)$ is the Laplace exponent corresponding to $\DD$ (see section~\ref{setup}). Multiplying \eqref{csubhat} by $\lambda$ and using that $\Psi(\lambda)\to0$ as $\lambda\to0$ and the final value theorem of Laplace transforms confirms the desired result that $c(x,t)\to {u}_{\text{ss}}(x)$ as $t\to\infty$. That is, the steady-state behavior of the subdiffusive solution is identical to the steady-state behavior of the diffusive solution. This result can also be seen from \eqref{ltl90}-\eqref{ltl91} in section~\ref{stability} above. 

\begin{figure}[t]
\centering
\includegraphics[width=1\linewidth]{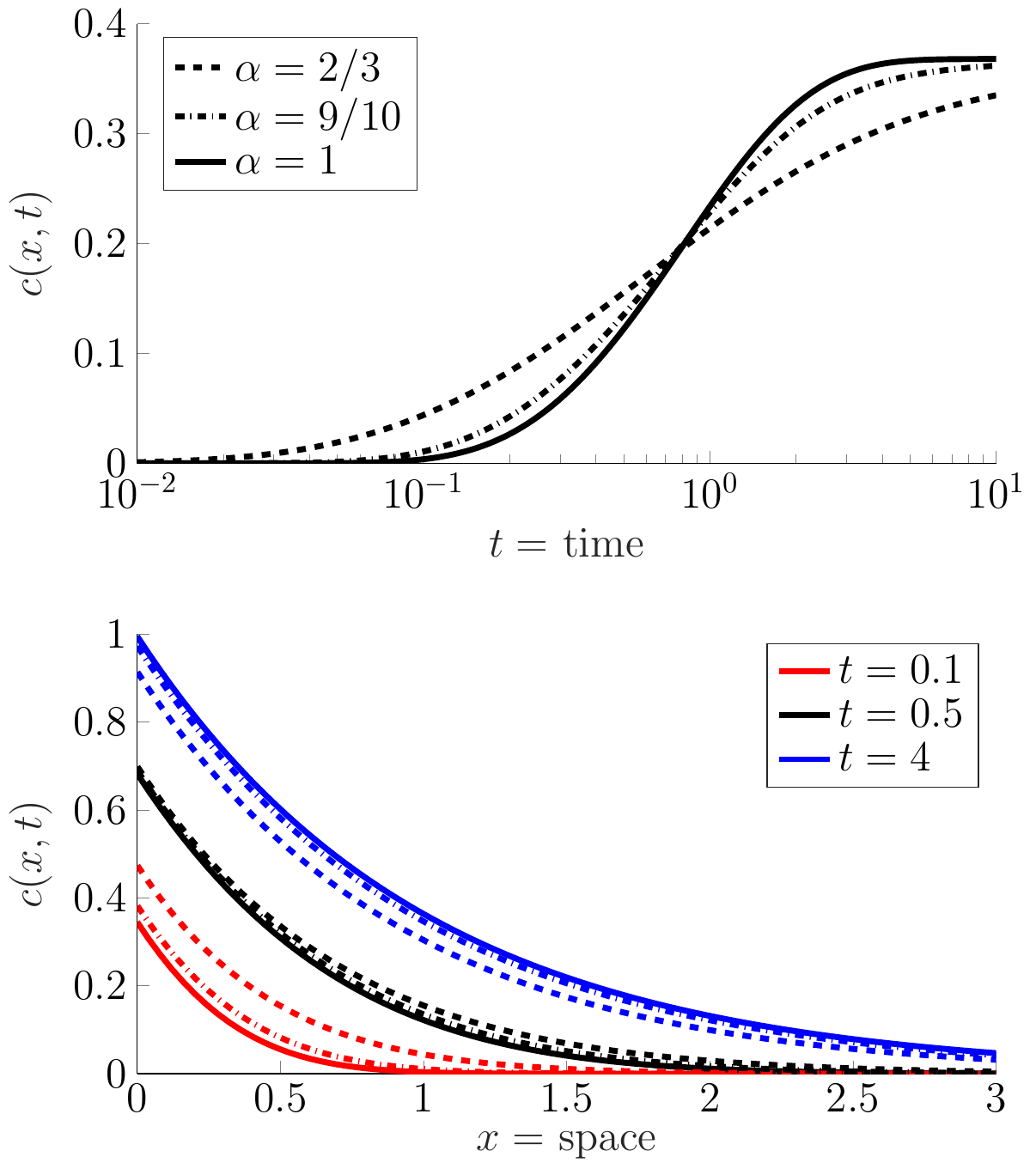}
\caption{\small Diffusive and subdiffusive gradient formation. The top panel plots the solution $c(x,t)$ to \eqref{second} as a function of time at $x=1$. The dashed curve is for $\alpha=2/3$, the dot-dashed curve is for $\alpha=9/10$, and the solid curve is normal diffusion ($\alpha=1$). The bottom panel plots $c(x,t)$ as a function of $x$ for $t=0.1,0.5,4$. The dashed, dot-dashed, and solid curves in the bottom panel correspond respectively to $\alpha=2/3$, $\alpha=9/10$, and $\alpha=1$, as in the top panel. See the text for more details.}
\label{figmorph}
\end{figure}

We are not able to analytically invert the Laplace transform in \eqref{csubhat}. Nevertheless, for a particular choice of $\Psi(\lambda)$, it straightforward to numerically invert \eqref{csubhat} to obtain $c(x,t)$. In Figure~\ref{figmorph}, we plot the protein concentration for the Laplace exponent,
\begin{align}\label{rlpsi}
\Psi(\lambda)=\lambda^{\alpha}, \quad \alpha\in(0,1],
\end{align}
which corresponds to the Riemann-Liouville operator $\DD=\D$ in \eqref{rl}. In the top panel in Figure~\ref{figmorph}, we plot the protein concentration as a function of time for $x=1$ and $\alpha=2/3$, $\alpha=9/10$, and $\alpha=1$ (the case $\alpha=1$ corresponds to normal diffusion). In the bottom panel in Figure~\ref{figmorph}, we plot the protein concentration as a function of space at a sequence of 3 time values. In these plots, we set $k$, $D$, and $\varphi$ to unity, and so the time, space, and concentrations can be interpreted as dimensionless.

From Figure~\ref{figmorph}, we see that (i) the protein concentration grows more quickly at early times for smaller values of $\alpha$ and (ii) the protein concentration grows more slowly at later times for {\black smaller} values of $\alpha$. In addition, the approach of the subdiffusive concentration $c(x,t)$ to the steady-state ${u}_{\text{ss}}(x)$ can be seen in Figure~\ref{figmorph}. However, we claim that the accumulation time formalism described above fails to quantify the timescale of this subdiffusive approach. To see this, define the subdiffusive accumulation time $\tau_{\text{sub}}(x)$ analogously to the diffusive accumulation time in \eqref{tau}, 
\begin{align*}
\tau_{\text{sub}}(x)
:=\int_{0}^{\infty}R_{\text{sub}}(x,t)\,\dd t,
\end{align*}
where the subdiffusive local relaxation function $R_{\text{sub}}(x,t)$ is defined analogously to \eqref{RRR},
\begin{align*}
R_{\text{sub}}(x,t)
=\frac{c(x,t)-c_{\text{ss}}(x)}{c(x,0)-c_{\text{ss}}(x)}
=1-\frac{c(x,t)}{{u}_{\text{ss}}(x)}.
\end{align*}
Using that $\tau_{\text{sub}}(x)$ can be written in terms of the Laplace transform of $R_{\text{sub}}(x,t)$ and using \eqref{csubhat}, we then obtain
\begin{align*}
\tau_{\text{sub}}(x)
&=\lim_{\lambda\to0+}\widehat{R_{\text{sub}}}(x,\lambda)
=\tau(x)\lim_{\lambda\to0+}\frac{\Psi(\lambda)}{\lambda}.
\end{align*}
Using the value $\Psi(\lambda)=\lambda^{\alpha}$ in \eqref{rlpsi} corresponding to the Riemann-Liouville fractional derivative, we obtain that the accumulation time is infinite if $\alpha\in(0,1)$,
\begin{align}\label{tauinf8}
\tau_{\text{sub}}(x)
=\infty.
\end{align}
The result in \eqref{tauinf8} is not surprising since $\tau_{\text{sub}}(x)$ is defined analogously to a mean first passage time and it is known that subdiffusive processes typically have infinite mean first passage times \cite{yuste2004}.


Summarizing, compared to normal diffusion, we see that this subdiffusive model of gradient formation yields a protein concentration that grows faster at early times and slower at later times. Further, while the subdiffusive concentration approaches the diffusive steady state at large time, the accumulation time formalism does not describe this timescale.





\subsection{\label{exswitch}Switching subdiffusivity}

A variety of systems in cell biology are characterized by macromolecules whose diffusivity randomly switches between two or more discrete values \cite{bressloffbook}. For example, AMPA receptors on the post-synaptic membrane switch between fast diffusive and stationary modes \cite{borgdorff2002}. Similarly, LFA-1 receptors switch between fast and slow diffusive modes \cite{das2009,slator2015}. Indeed, the prevalence of such processes in cell biology is evidenced by the various statistical methods that have been created to study single particle tracking data and detect fluctuations in diffusion coefficients \cite{das2009,koo2016,monnier2013,montiel2006,persson2013,slator2018,slator2015}.

Switching diffusion coefficients often model (a) binding/unbinding of the diffusing particle to other molecules that alter its mobility or (b) switching conformations, with distinct mobilities corresponding to the effective sizes of the conformations \cite{cairo2006,wu2018,grebenkov2019}. If the motion of the particles is subdiffusive, and the factors causing the subdiffusion similarly hamper the transitions between states, then the spatiotemporal evolution of the particle population could be modeled by an equation of the form in \eqref{form0}. To illustrate, consider
\begin{align}\label{easyexample}
\begin{split}
&\frac{\partial}{\partial t}
\begin{pmatrix}
c_{0}\\
c_{1}
\end{pmatrix}
=\DD
\Delta
\begin{pmatrix}
K_{0}c_{0}\\
K_{1}c_{1}
\end{pmatrix}
+\DD\begin{pmatrix}
-\lambda_{0} & \lambda_{1}\\
\lambda_{0} & -\lambda_{1}
\end{pmatrix}\begin{pmatrix}
c_{0}\\
c_{1}
\end{pmatrix},
\end{split}
\end{align}
which models a population of particles that switch between two states and subdiffuse in state $j\in\{0,1\}$ with generalized diffusivity $K_{j}$. If $\DD=\D$, then section~\ref{stochrep} shows that the dwell times in each state have the Mittag-Leffler distribution (see \eqref{ml}).

Further, section~\ref{stochrep} shows that the stochastic state of an individual particle following \eqref{easyexample} is given by
\begin{align*}
(Y(t),J(t))
:=(X(S(t)),I(S(t)))\in\overline{\domain}\times\{0,1\},
\end{align*}
where $S(t)$ is the inverse of a subordinator $T(s)$ with L\'{e}vy exponent given by the reciprocal of the Laplace transform of the memory kernel in $\DD$ (see section~\ref{setup}), $I(s)\in\{0,1\}$ is a two-state Markov jump process with jump rates $\lambda_{0},\lambda_{1}$, and $X(s)$ follows the switching SDE,
\begin{align*}
\dd X(s)=\sqrt{2K_{I(s)}}\,\dd W(s).
\end{align*}

\begin{figure}
\centering
\includegraphics[width=1\linewidth]{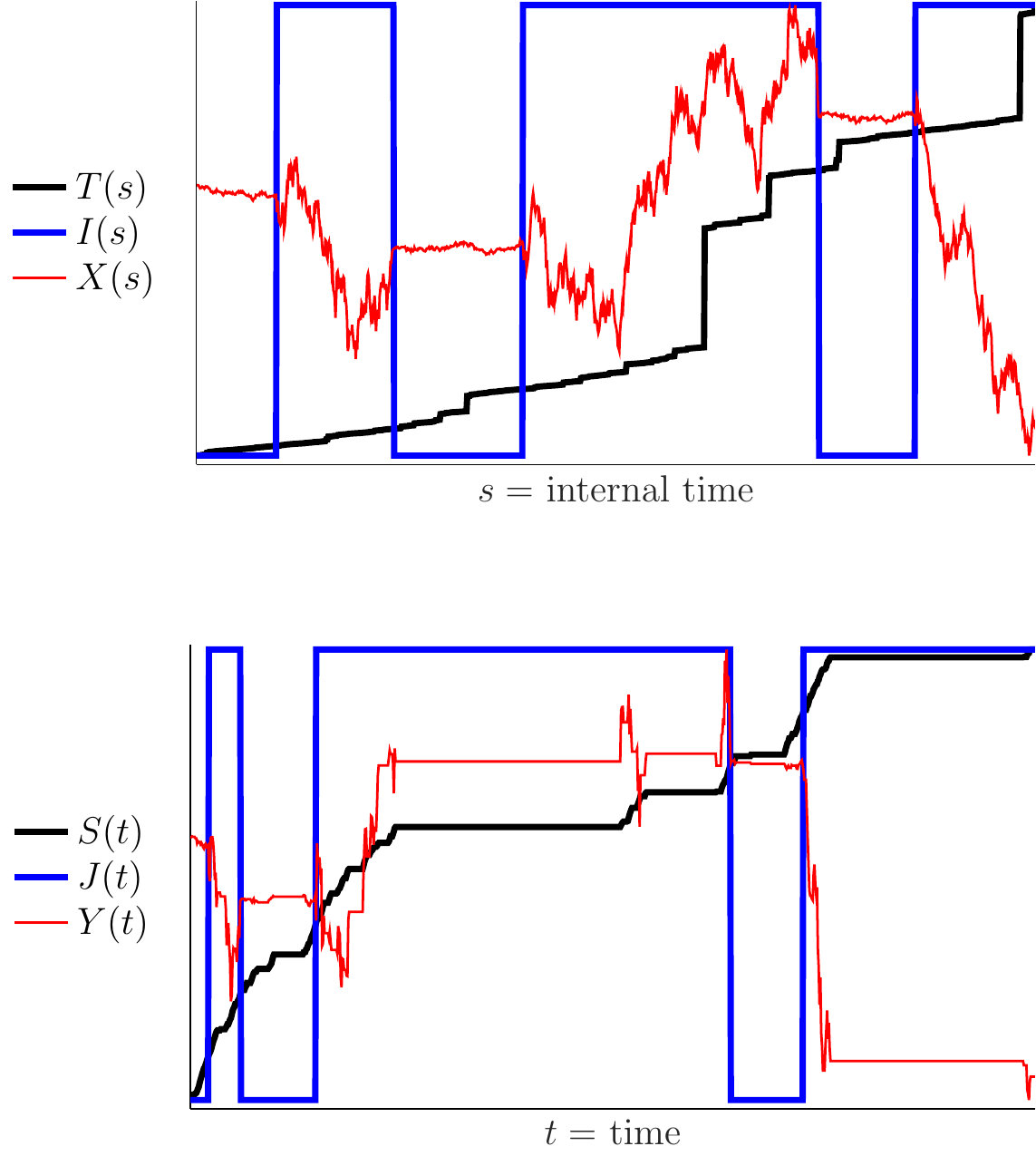}
\caption{\small Switching subdiffusivity. In the top panel, we plot $T(s)$, $I(s)$, and $X(s)$ as functions of the internal time $s$. In the bottom panel, we plot $S(t)$, $J(t)$, and $Y(t)$ as functions of time $t$. See the text for details.}
\label{figsw}
\end{figure}

In Figure~\ref{figsw}, we plot a realization of $(Y(t),J(t))$ and the corresponding realizations of $S(t)$, $X(s)$, $I(s)$, and $T(s)$ by employing the method described in section~\ref{stochsim} above. In this plot, we take the fractional operator to be the Riemann-Liouville derivative, $\DD=\D$, with $\alpha=3/4$, and set $\lambda_{0}=\lambda_{1}$ and $K_{1}/K_{0}=100$ so that the process moves much more quickly in state 1 compared to state 0.

In the top panel of Figure~\ref{figsw}, we plot $T$, $I$, and $X$ as functions of the internal time $s$. Notice that $T(s)$ is an increasing process which occasionally takes large jumps. Notice also that $X(s)$ diffuses much faster when $I(s)=1$ compared to when $I(s)=0$. In the bottom panel, we plot $S$, $J$, and $Y$ as functions of time $t$. Notice that jumps in $T$ correspond to flat periods or ``pauses'' in $S$. Notice also that both $J$ and $Y$ pause when $S$ pauses. In particular, though $I(s)$ switches states at exponentially distributed times, the pauses in $J(t)$ induced by $S(t)$ make $J(t)$ switch states at Mittag-Leffler distributed times (see \eqref{ml}). Furthermore, notice that if $Y$ is not paused, then it moves much more quickly when $J(t)=1$ compared to when $J(t)=0$. We note that we have shifted and scaled the vertical axes in Figure~\ref{figsw} so that the various curves fit on the same plots.



\subsection{Space-dependent switching and gradient formation}

In the example in section~\ref{exswitch} above, the particles switch states at rates that are independent of their spatial position. It was recently shown that space-dependent switching can induce the formation of protein concentration gradients inside a single cell \cite{wu2018}. This mechanism of gradient formation is particularly notable since the more classical mechanism involving diffusion away from a localized source and subsequent degradation (as in section~\ref{morph} above) typically fails at subcellular length scales \cite{kholodenko2009, howard2012} 

This situation has been modeled by \cite{wu2018, PB13}
\begin{align}\label{rd}
\begin{split}
&\frac{\partial}{\partial s}
\begin{pmatrix}
{{u}}_{0}\\
{{u}}_{1}
\end{pmatrix}
=
\Delta
\begin{pmatrix}
K_{0}{{u}}_{0}\\
K_{1}{{u}}_{1}
\end{pmatrix}
+\frac{1}{\eps}\begin{pmatrix}
-\lambda_{0}(x) & \lambda_{1}(x)\\
\lambda_{0}(x) & -\lambda_{1}(x)
\end{pmatrix}\begin{pmatrix}
{{u}}_{0}\\
{{u}}_{1}
\end{pmatrix},
\end{split}
\end{align}
where $u_{j}(x,s)$ is the concentration of molecules in state $j\in\{0,1\}$ at time $s\ge0$ at position $x$ in the finite interval $[0,L]$. Notice that the rate $\lambda_{j}(x)$ of leaving state $j$ depends on the current spatial position. In \eqref{rd}, a small dimensionless parameter $\eps>0$ has been introduced to model switching that occurs on a much faster timescale than gradient formation. It was shown in \cite{PB13} that if $\eps\ll1$, then the large time total concentration ${u}(x):=\lim_{s\to\infty}{{u}}_{0}(x,s)+{{u}}_{1}(x,s)$ is proportional to
\begin{align}\label{uform}
{{u}}(x)
\propto\Big(\frac{\lambda_{1}(x)}{\lambda_{0}(x)+\lambda_{1}(x)}K_{0}+\frac{\lambda_{0}(x)}{\lambda_{0}(x)+\lambda_{1}(x)}K_{1}\Big)^{-1},
\end{align}
assuming no flux boundary conditions for ${{u}}_{j}$ at $x=0,L$. The form in \eqref{uform} means that molecules concentrate in regions where they are more likely to be in a slower state. This point is related to a fairly subtle point regarding It\'{o} versus Stratonovich stochastic integration \cite{PB8, PB10}.

Given the ubiquity of subdiffusive motion inside cells, it is natural to ask if this same mechanism for gradient formation exists for subdiffusion. If the reactions causing the transitions between states is subdiffusion-limited, then the concentrations $c_{0}(x,t)$ and $c_{1}(x,t)$ can be modeled by the equations in \eqref{rd} with the operator $\DD$ applied to the righthand side. Our analysis in section~\ref{solution} thus shows that the subdiffusive concentrations are $c_{j}(x,t)=\E[{{u}}_{j}(x,S(t))]$. It then follows from our analysis in section~\ref{stability} that the large time total subdiffusive concentration is exactly given by \eqref{uform}, which shows that this mechanism of intracellular gradient formation extends to subdiffusive motion.

\subsection{FRAP experiments}

Fluorescence recovery after photobleaching (FRAP) is a commonly used experimental method for studying binding interactions in cells \cite{lippincott2018, ponce2020}. Though subdiffusion is widely observed in cells, the vast majority of mathematical models of FRAP experiments assume that the molecules move by normal diffusion (but see the work of Yuste et al.\ \cite{yuste2014} for a notable exception).

In the case of normal diffusion, the influential work of Sprague et al.\ \cite{sprague2004} considers the following linear reaction-diffusion equations describing a FRAP system in a two-dimensional disk,
\begin{align}\label{spr}
\begin{split}
\frac{\partial {{u}}_{0}}{\partial s}
&=D\Big(\frac{1}{r}\frac{\partial}{\partial r}+\frac{\partial^{2}}{\partial r^{2}}\Big){{u}}_{0}-\kon {{u}}_{0}+\koff {{u}}_{1},\\
\frac{\partial {{u}}_{1}}{\partial s}
&=\kon {{u}}_{0}-\koff {{u}}_{1},
\end{split}
\end{align}
for free (respectively bound) proteins ${{u}}_{0}(r,s)$ (respectively ${{u}}_{1}(r,s)$) at radius $r\in(0,\rho)$ at time $s\ge0$. In order to compare to experimental data, one calculates the so-called FRAP curve, which is the sum ${{u}}_{0}+{{u}}_{1}$ averaged over the disk,
\begin{align}\label{frap}
\text{frap}(s)
:=\frac{2}{\rho^{2}}\int_{0}^{\rho}\Big({{u}}_{0}(r,s)+{{u}}_{1}(r,s)\Big)r\,\dd r.
\end{align}
While an explicit formula for \eqref{frap} is unknown, Sprague et al. \cite{sprague2004} found the following exact formula for its Laplace transform, 
\begin{align}\label{fraplt}
\begin{split}
&\widehat{\text{frap}}(\lambda)
=\frac{1}{\lambda}-\frac{\kon}{(\lambda+\koff)(\kon+\koff)}\\
&-\frac{\koff}{\lambda(\kon+\koff)}\Big(1-2K_{1}(q\rho)I_{1}(q\rho)\Big)\Big(1+\frac{\kon}{\lambda+\koff}\Big),
\end{split}
\end{align}
where $I_{1}$ and $K_{1}$ are modified Bessel functions of the first and second kind and
\begin{align*}
q
=\sqrt{\frac{\lambda}{D}\Big(1+\frac{\kon}{\lambda+\koff}\Big)}.
\end{align*}
Note that \eqref{fraplt} has been normalized so that it yields $\lim_{s\to\infty}\text{frap}(s)=1$. The Laplace transform \eqref{fraplt} can be inverted numerically to yield the FRAP curve \eqref{frap} and then be compared to experimental data \cite{sprague2004}.

We can extend these results to the case that the proteins move by subdiffusion and the reactions are subdiffusion-limited. In particular, suppose the subdiffusion is modeled with the fractional operator $\DD$ in \eqref{DD}. Let $\text{frap}_{\text{sub}}(t)$ denote the subdiffusive FRAP curve defined as in \eqref{frap}, but where ${u}_{0}$ and ${u}_{1}$ are replaced by $c_{0}$ and $c_{1}$ which satisfy \eqref{spr} with $\DD$ applied to the righthand sides. Theorem~\ref{laplacespace} then implies that the Laplace transform of the subdiffusive FRAP curve is given explicitly in terms of \eqref{fraplt},
\begin{align}\label{fraplt2}
\widehat{\text{frap}_{\text{sub}}}(\lambda)
=\frac{\Psi(\lambda)}{\lambda}\widehat{\text{frap}}(\Psi(\lambda)),\quad \lambda>0,
\end{align}
where $\Psi(\lambda)$ corresponds to $\DD$ (see section~\ref{setup}). As above, \eqref{fraplt2} can be inverted numerically to yield the subdiffusive FRAP curve.

\begin{figure}[t]
\centering
\includegraphics[width=1\linewidth]{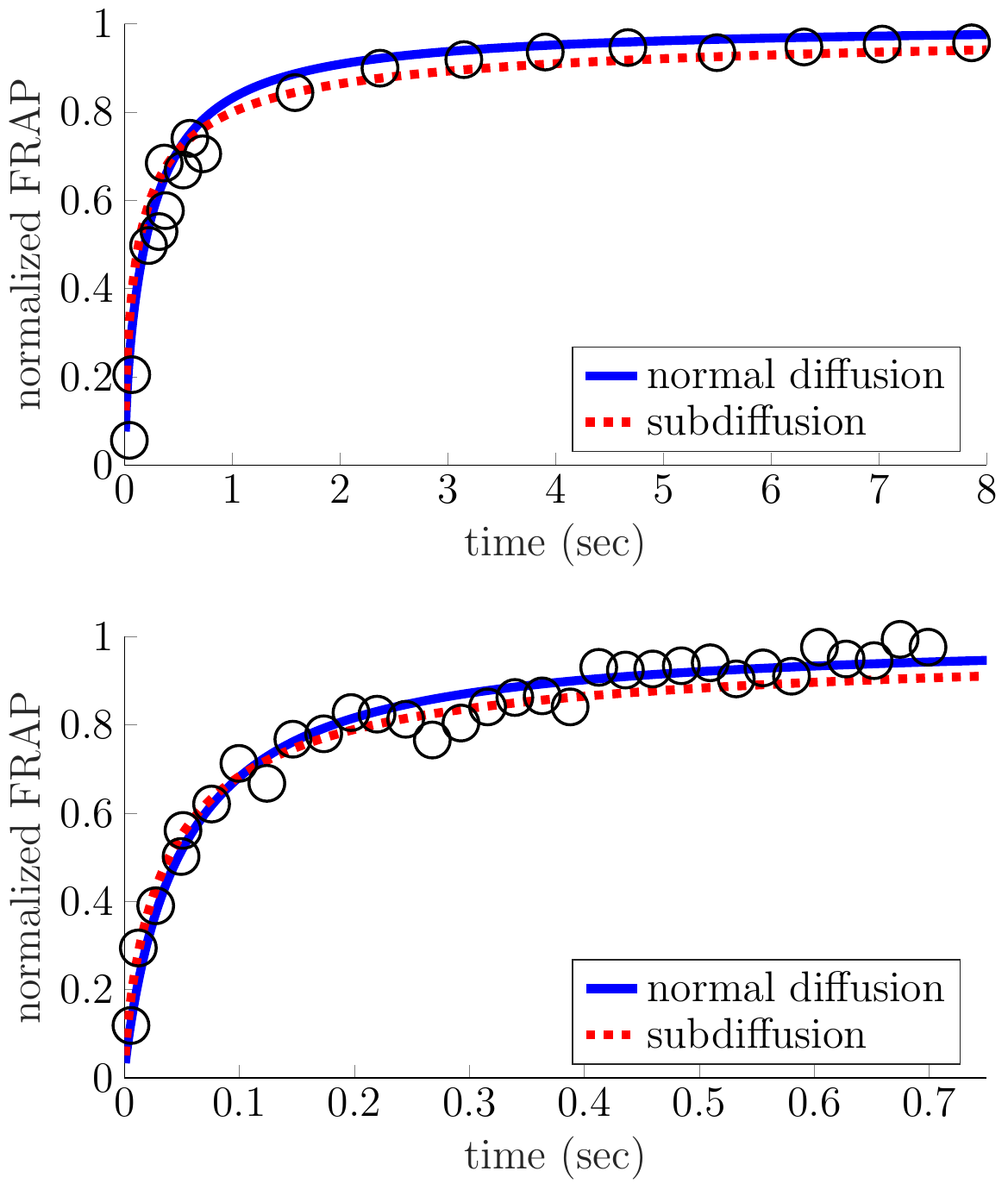}
\caption{\small FRAP curves for normal diffusion (blue solid) and subdiffusion (red dashed) can fit experimental data (black circles) of \cite{sprague2004}. See the text for details.}
\label{figfrap}
\end{figure}

In Figure~\ref{figfrap}, we plot the diffusive FRAP curve and the subdiffusive FRAP curve as functions of time. The circles in the top panel in Figure~\ref{figfrap} are experimental data points from Figure 5E in \cite{sprague2004}. Similarly, the circles in the bottom panel in Figure~\ref{figfrap} are data points from Figure 5F in \cite{sprague2004}. Figure~\ref{figfrap} shows that the subdiffusion-limited FRAP model described above can fit this experimental data of \cite{sprague2004}. \textcolor{black}{In particular, the subdiffusion-limited FRAP model and the normal diffusion FRAP model fit the experimental data of \cite{sprague2004} roughly equally well. Hence, this data alone cannot distinguish between the two models.} This figure follows Figures 1 and 2 in \cite{yuste2014} that showed that a different subdiffusive FRAP model can also fit this experimental data of \cite{sprague2004} {\black roughly equally well as the normal diffusion model}.

The parameters used in Figure~\ref{figfrap} are as follows. In Figure~\ref{figfrap}, the radius is $\rho=1.1\,\mu\text{m}$ in the top panel and $\rho=0.5\,\mu\text{m}$ in the bottom panel. For diffusive FRAP (blue solid curves), we take $\kon=400\,\text{sec}^{-1}$, $\koff=78.6\,\text{sec}^{-1}$, and $D=9.2\,\mu\text{m}^{2}\text{sec}^{-1}$ in both panels. For the subdiffusive FRAP (red dashed curves), we take the fractional operator to be the Riemann-Liouville operator $\DD=\D$ with $\alpha=0.75$, and set $\kon=750\,\text{sec}^{-\alpha}$, $\koff=17\,\text{sec}^{-\alpha}$, and $D=82\,\mu\text{m}^{2}\text{sec}^{-\alpha}$ in both panels. The parameters for the diffusive FRAP curves were used in Figure 5F in \cite{sprague2004} (slightly different parameters were used in Figure 5E in \cite{sprague2004}, but we use the same parameters in both panels).

\section{Fractional ODEs\label{sectionode}}

Our results hold in significant generality, essentially requiring only that the operator $\mathcal{A}$ commutes with temporal operators (see \eqref{commute1}-\eqref{commute3}). Indeed, the equations need not even involve the spatial variable $x$, and can instead be a system of fractional ODEs. Fractional ODEs have been used to model a variety of systems, including pharmacokinetics \cite{dokoumetzidis2009} and the spread of an infectious disease through a population \cite{angstmann2017}.

\subsection{Solution}

Consider the affine fractional ODEs, 
\begin{align}\label{nonsp}
\frac{\dd}{\dd t}\c(t)
=\DD(R\c(t)+\r),
\end{align}
where $\c(t)=({c_{i}}(t))_{i=1}^{n}\in\R^{n}$ is a time-dependent solution vector and $R\in\R^{n\times n}$ is a matrix and $\r\in\R^{n}$ is a vector. In this case, section~\ref{solution} yields the relations
\begin{align}
\c(t)
&=\E[\w(S(t))],\label{later88}\\
\widehat{\c}(\lambda)
&=\frac{\Psi(\lambda)}{\lambda}\widehat{\w}(\Psi(\lambda)),\nonumber
\end{align}
as in \eqref{soln}-\eqref{solnlp}, where $\w$ satisfies the ODE
\begin{align}\label{wode}
\frac{\dd}{\dd s}\w(s)
=R\w(s)+\r,
\end{align}
with $\w(0)=\c(0)\in\R^{n}$.

The solution $\w(s)$ to \eqref{wode} is of course
\begin{align*}
\w(s)
&=e^{Rs}\w(0)+\int_{0}^{s}e^{R(s-\sigma)}\r\,\dd \sigma\\
&=\sum_{k=0}^{\infty}\frac{R^{k}s^{k}}{k!}\w(0)+\sum_{k=0}^{\infty}\frac{R^{k}s^{k+1}}{(k+1)!}\r.
\end{align*}
Hence, \eqref{later88} yields the following explicit formula for the fractional solution in terms of the moments of $S(t)$,
\begin{align}\label{series}
\c(t)
&=\sum_{k=0}^{\infty}\frac{R^{k}\E[(S(t))^{k}]}{k!}\c(0)+\sum_{k=0}^{\infty}\frac{R^{k}\E[(S(t))^{k+1}]}{(k+1)!}\r.
\end{align}

In the case that the fractional operator is the Riemann-Liouville derivative, $\DD=\D$, we have that \cite{piryatinska2005}
\begin{align}\label{rlsk}
\E[(S(t))^{k}]
=\frac{t^{\alpha k}k!}{\Gamma(1+\alpha k)}.
\end{align}
Plugging \eqref{rlsk} into \eqref{series} yields a formula for $\c(t)$ that agrees with a recent result of Duan \cite{duan2018}.

\subsection{\textcolor{black}{Steady-states and stability}}

Equations~\eqref{ltl90}-\eqref{ltl91} in section~\ref{stability} above show that if $\w$ approaches a finite limit at large time, then $\c$ must also approach this same limit at large time. Furthermore, the results of section~\ref{stability} yield that if a nonlinear integer order ODE is linearly stable, then the corresponding nonlinear fractional order ODE is also linearly stable.

However, we caution that the stability of an integer order ODE cannot be inferred from the stability of the corresponding fractional ODE. Indeed, if the fractional operator is the Riemann-Liouville operator, $\DD=\D$, then it is known \cite{matignon1996, li2011} that the origin is asymptotically stable for the linear fractional ODE
\begin{align}\label{fcnice}
\frac{\dd}{\dd t}\c(t)
=\D R\c(t),
\end{align}
if and only if 
\begin{align}\label{argcond}
|\text{Arg}(\nu)|
>\frac{\alpha\pi}{2},
\end{align}
for every eigenvalue $\nu\in\mathbb{C}$ of $R\in\mathbb{R}^{n\times n}$, where $\text{Arg}(\nu)\in(-\pi,\pi]$ denotes the principal argument of $\nu$. Notice that \eqref{argcond} generalizes the classical result for integer order ODEs with $\alpha=1$. Hence, if $R$ satisfies \eqref{argcond} and 
\begin{align*}
|\text{Arg}(\nu)|
<\frac{\pi}{2},
\end{align*}
for some $\nu\in\mathbb{C}$, then the solution to \eqref{fcnice} vanishes but the solution to the corresponding integer equation diverges.

\subsection{Stochastic representation}

The stochastic representation of section~\ref{stochrep} above still holds in the non-spatial case of \eqref{nonsp} if $\r=0$ and $R$ is the forward operator for a continuous-time Markov chain (as in section~\ref{stochrep}). In this case, if $\{I(s)\}_{s\ge0}$ is a continuous-time Markov chain with forward operator $R$, then the probability distribution of $J(t):=I(S(t))$ satisfies \eqref{nonsp} with $\r=0$. We note that this connection between fractional order and integer order Markov chains was investigated in \cite{repin2000, jumarie2001, laskin2003, mainardi2004, mainardi2007, uchaikin2008, beghin2009, beghin2010, meerschaert2011} in the case that $\DD$ is the Riemann-Liouville derivative and $\{I(s)\}_{s\ge0}$ is a Poisson process.

\section{Discussion}

We have analyzed subdiffusion-limited mesoscopic equations describing a reaction-subdiffusion system in a general mathematical setting, under the assumption that the reactions are affine. We have shown that the solution to this fractional system is the expectation of a random time change of the corresponding integer order system. This result yielded (i) a simple algebraic relation between the fractional solution and the integer order solution in Laplace space, (ii) a sufficient condition for the {\black linear} stability of fractional equations with nonlinear reactions in terms of the {\black linear} stability of the corresponding integer order equations, and (iii) the exact microscopic description of single molecules corresponding to these mesoscopic equations and a numerical method for their stochastic simulation.

These results extend previous results for subdiffusive systems with no reactions. Barkai \cite{barkai2001} found the solution to a fractional Fokker-Planck equation in $\R$ in terms of the solution to the corresponding integer order Fokker-Planck equation in the case that the fractional operator is the Riemann-Liouville derivative. Magdziarz \cite{magdziarz2009} found the stochastic representation for such fractional Fokker-Planck equations in $\R$ when the fractional operator involves a general memory kernel. This was further generalized in \cite{magdziarz2016} by Magdziarz and Zorawik. In addition, fractional Fokker-Planck equations in $\R^{d}$ with general memory kernels were considered by Carnaffan and Kawai \cite{carnaffan2017}. Similar stochastic representations of solutions to fractional equations have been found in \cite{baeumer2001, meerschaert2004, chen2017, du2020}. {\black An additional related work is that of Yadav and Horsthemke \cite{yadav2006}, which derived a different class of reaction-subdiffusion equations and analyzed their linear stability.} 

{\black
An alternative to the subdiffusion-limited model considered in the present work is the activation-limited model \cite{nepomnyashchy2016}. In contrast to subdiffusion-limited reactions, activation-limited reaction rates are unaffected by subdiffusive processes. To illustrate in a simple example, consider a chemical which (i) subdiffuses in $\R^{d}$ with generalized diffusivity $K>0$ and (ii) switches between $n$ discrete states according to a constant reaction rate matrix $R\in\R^{n\times n}$. Let $\c(x,t)$ denote the vector of these $n$ chemical concentrations. In the subdiffusion-limited model, $\c$ evolves according to
\begin{align}\label{subl}
\frac{\partial}{\partial t}\c
=\DD(K\Delta \c+R\c),\quad x\in\R^{d},\,t>0,
\end{align}
where $\DD$ is as in section~\ref{solution}. In the activation-limited model, $\c$ evolves according to \cite{henry2006, sokolov2006, schmidt2007, langlands2008, lawley2020sr1}
\begin{align}\label{actl}
\frac{\partial}{\partial t}\c
=e^{Rt}\DD e^{-Rt}K\Delta \c+R\c,\quad x\in\R^{d},\,t>0,
\end{align}
where $e^{\pm Rt}$ denotes the matrix exponential.

As we showed in section~\ref{stochrep}, \eqref{subl} describes individual molecules whose discrete state dynamics depend on their subdiffusive behavior. Indeed, molecules following \eqref{subl} cannot switch state when they are in a subdiffusive ``pause,'' and this forces the random time between switches to have a Mittag-Leffler distribution (see \eqref{ml}). In contrast, it was recently proven in \cite{lawley2020sr1} that \eqref{actl} is a direct consequence of the independence of the discrete state and subdiffusive motion, and thus the molecules switch states at exponentially distributed times.

Differences between \eqref{subl} and \eqref{actl} can also be seen by examining their solutions. Assume an initial condition $\c(x,0)=u_{0}(x)\v$ for some function $u_{0}:\R^{d}\mapsto\R$ and some vector $\v\in\R^{n}$. If $u(x,s)\in\R$ satisfies the single-component normal diffusion equation,
\begin{align*}
\frac{\partial}{\partial s}u
&=K\Delta u,\quad x\in\R^{d},\,s>0\\
u
&=u_{0},\quad x\in\R^{d},\,s=0,
\end{align*}
then it follows from the analysis in section~\ref{solution} that the solution to \eqref{subl} is
\begin{align}\label{sublsoln}
\c(x,t)=\E[u(x,S(t))e^{RS(t)}]\v,
\end{align}
where $S(t)$ is as in section~\ref{solution}. In contrast, it follows from section~\ref{solution} above and the results of \cite{lawley2020sr1} that the solution to \eqref{actl} is
\begin{align}\label{actlsoln}
\c(x,t)=\E[u(x,S(t))]e^{Rt}\v.
\end{align}
Since the matrix exponential describes molecular reactions, it is evident that subdiffusion modifies the reactions in \eqref{sublsoln} (since the matrix exponential is subordinated by $S(t)$), whereas the reactions are unaffected by subdiffusion in \eqref{actlsoln}.
}

We used our results to explore how subdiffusion modifies several models in cell biology. The Laplace space relation we found between solutions to fractional and integer order equations allowed us to quickly convert results from diffusive models to subdiffusive models. Our results suggest that mechanisms for gradient formation which have been formulated for diffusive molecules extend to subdiffusive molecules. In addition, it is interesting that our subdiffusive FRAP model closely fits data from FRAP experiments \cite{sprague2004} {\black (the fit is roughly the same as the normal diffusion FRAP model)}. This parallels the work of Yuste et al.\ \cite{yuste2014}, which found similar results for a different subdiffusive FRAP model. 

{\black 
More generally, subordination methods (i.e.\ random time changes) similar to the one employed in the present work have been used to understand stochastic phenomena in many physical problems. For example, a variety of systems exhibit ``anomalous yet Brownian'' diffusion, which is defined by a linear mean-squared displacement with non-Gaussian increments \cite{wang2009}. Such systems have been modeled by diffusing diffusivity \cite{chubynsky2014}, which is equivalent to a certain subordination of diffusion \cite{chechkin2017}. In addition, subdiffusion and superdiffusion have been modeled by grey Brownian motion \cite{schneider1990, mura2008}, which can be represented in terms of a subordination of more classical processes \cite{silva2015}.
}

We also applied our results to fractional ODEs. Our work extends recent solution formulas for fractional ODEs \cite{duan2018} to more general fractional operators. In addition, our work complements and extends some previous work on fractional Poisson processes \cite{repin2000, jumarie2001, laskin2003, mainardi2004, mainardi2007, uchaikin2008, beghin2009, beghin2010, meerschaert2011}. 

While our results are formulated in significant  mathematical generality, we did assume that the reactions are affine functions, {\black which is perhaps the main limitation of our results. Some} previous studies considered models with nonlinear reactions (often mass action kinetics) \cite{yuste2004b, kosztolowicz2006, kosztolowicz2008, kosztolowicz2013, nec2010, nec2013, nepomnyashchy2013}. Hence, further investigating the relationship between fractional and integer order equations involving nonlinearities remains an important direction for future work.


\begin{acknowledgments}
The author gratefully acknowledges support from the National Science Foundation (DMS-1944574, DMS-1814832, and DMS-1148230).
\end{acknowledgments}

{\black
\section{Appendix}}

\textcolor{black}{In this Appendix, we give the proof of Theorem~\ref{laplacespace}. The proof} relies on the following lemma. We write $\textcolor{black}{\tau
=_{\dist}\exp(\lambda)}$ to denote that $\tau$ is exponentially distributed with rate $\lambda>0$, which means $\P(\tau>t)=e^{-\lambda t}$ for each $t>0$.

\begin{lemma}\label{softau}
If $\tau=_{\dist}\exp(\lambda)$ and is independent of $T$, then
\begin{align*}
S(\tau)
=_{\dist}\exp(\Psi(\lambda)).
\end{align*}
\end{lemma}

\begin{proof}[Proof of Lemma~\ref{softau}]
Fix $s>0$. Using the definition of $S(t)$ in \eqref{S} and the independence of $T$ and $\tau$, conditioning on the value of $\tau$ gives
\begin{align}\label{first}
\begin{split}
\P(S(\tau)> s)
&=\int_{0}^{\infty}F(t)\lambda e^{-\lambda t}\,\dd t,
\end{split}
\end{align}
where $F(t):=\P(T(s)\le t)$ and we have used \eqref{ctsrv}. Integrating by parts in \eqref{first} yields
\begin{align}\label{rst}
\int_{0}^{\infty}F(t)\lambda e^{-\lambda t}\,\dd t
=
\int_{0}^{\infty} e^{-\lambda t}\dd F(t),
\end{align}
since $\lim_{t\to\infty}e^{-\lambda t}F(t)=0$ and $F(0)=\P(T(s)\le0)=0$ by \eqref{ctsrv} since $s>0$. Now, \eqref{le} implies that the Riemann-Stieltjes integral in the righthand side of \eqref{rst} is
\begin{align}\label{usingle}
\int_{0}^{\infty}e^{-\lambda t}\dd F(t)
=\E[e^{-\lambda T(s)}]
=e^{-s\Psi(\lambda)}.
\end{align}
Combining \eqref{first}-\eqref{usingle} completes the proof.
\end{proof}

The proof of Theorem~\ref{laplacespace} follows quickly from Lemma~\ref{softau}.

\begin{proof}[Proof of Theorem~\ref{laplacespace}]
The Laplace transform of $\c(t)$ is
\begin{align*}
\widehat{\c}(\lambda)
&=\int_{0}^{\infty}e^{-\lambda t}\c(t)\,\dd t
=\int_{0}^{\infty}e^{-\lambda t}\E[{\u}(S(t))]\,\dd t\\
&=\E\int_{0}^{\infty}e^{-\lambda t}{\u}(S(t))\,\dd t
=\frac{1}{\lambda}\E[{\u}(S(\tau))],
\end{align*}
where $\tau=_{\dist}\exp(\lambda)$ is independent of $S$ (the assumption \eqref{tonelli} and the theorems of Tonelli and Fubini ensure the validity of exchanging $\E$ with the integral). Therefore, if $\sigma
=_{\dist}\exp(\Psi(\lambda))$, then Lemma~\ref{softau} implies that
\begin{align}\label{relalp}
\begin{split}
\widehat{\c}(\lambda)
=\frac{1}{\lambda}\E[{\u}(\sigma)]
&=\frac{1}{\lambda}\int_{0}^{\infty}\Psi(\lambda)e^{-\Psi(\lambda)t}{\u}(t)\,\dd t\\
&=\frac{\Psi(\lambda)}{\lambda}\widehat{{\u}}(\Psi(\lambda)),
\end{split}
\end{align}
which completes the proof.
\end{proof}


\bibliography{library.bib}

\end{document}